\documentclass[letterpaper, 10 pt, conference]{ieeeconf} 
\usepackage{xpatch, amsmath,amssymb,bm,mathtools}
\usepackage[capitalize]{cleveref}
\usepackage{xcolor}
\usepackage{fp, tikz, pgfplots}
\usepackage{siunitx}
\usepackage{url}
\usepackage{booktabs} 
\usepackage{tabularx}

\usetikzlibrary{arrows}
\usetikzlibrary{shapes}
\usetikzlibrary{patterns}
\usetikzlibrary{decorations.pathmorphing}
\usetikzlibrary{decorations.pathreplacing}
\usetikzlibrary{decorations.shapes}
\usetikzlibrary{decorations.text}
\usetikzlibrary{shapes.misc}
\usetikzlibrary{decorations.markings}
\usetikzlibrary{arrows.meta}
\usetikzlibrary{backgrounds}
\usepgfplotslibrary{external}
\usepgfplotslibrary{fillbetween}
\usepgfplotslibrary{statistics}
\usepackage[norelsize,ruled,linesnumbered,noend]{algorithm2e}
\usepackage{setspace}

\IEEEoverridecommandlockouts                              

\newtheorem{assumption}{Assumption}
\newtheorem{lemma}{Lemma}
\newtheorem{remark}{Remark}
\newtheorem{theorem}{Theorem}

\newcommand{\approptoinn}[2]{\mathrel{\vcenter{
			\offinterlineskip\halign{\hfil$##$\cr
				#1\propto\cr\noalign{\kern2pt}#1\sim\cr\noalign{\kern-2pt}}}}}

\makeatletter
\xpatchcmd{\proof}{\topsep0\p@\@plus0\p@\relax}{}{}{}
\makeatother

\pgfplotsset{width=\columnwidth, compat = 1.13, 
	height = 3.cm, grid= major, 
	legend cell align = left, ticklabel style = {font=\scriptsize},
	every axis label/.append style={font=\small},
	legend style = {font=\scriptsize\sffamily},title style={yshift=-7pt, font = \small} }

\definecolor{green}{rgb}{0.4660,0.6740,0.1880}
\definecolor{red}{rgb}{0.5,0,0}

\tikzset{cross/.style={cross out, draw=black, minimum size=2*(#1-\pgflinewidth), inner sep=0pt, 
		outer sep=0pt},
	cross/.default={1pt}}

\title{\LARGE \bf Networked Online Learning for Control of Safety-Critical Resource-Constrained Systems based on Gaussian Processes}

\author{Armin Lederer, Mingmin Zhang, Samuel Tesfazgi and Sandra Hirche
	\thanks{The authors are with the Department of Electrical and Computer Engineering, Technical University of Munich, 80333 Munich, Germany
		{\tt\small [armin.lederer, mingmin.zhang,
			samuel.tesfazgi, hirche]@tum.de}}%
}

\begin{document}
	
	\setlength{\textfloatsep}{11pt}
	\setlength{\floatsep}{2pt}
	\setlength{\abovedisplayskip}{5pt}
	\setlength{\belowdisplayskip}{5pt}
	
	\maketitle
	\thispagestyle{empty}
	\pagestyle{empty}

	\begin{abstract}
		Safety-critical technical systems operating in unknown 
		environments require the ability to quickly adapt their behavior, which can be achieved in control by inferring a model online from the data stream generated during operation. Gaussian process-based learning is particularly well suited for safety-critical applications as it ensures bounded prediction errors. 
		While there exist computationally efficient approximations for online inference, these approaches lack guarantees for the prediction error and have high memory requirements, and are therefore not applicable to safety-critical systems with tight memory constraints.
		In this work, we propose a novel networked online learning approach based on Gaussian process regression, which addresses the issue of limited local resources by employing remote data management in the cloud. 
		Our approach formally guarantees a bounded tracking error with high probability, which is exploited to identify the most relevant data to achieve a certain control performance. We further propose an effective data transmission scheme between the local system and the cloud taking bandwidth limitations and time delay of the transmission channel into account. The effectiveness of the proposed method is successfully demonstrated in a simulation.\looseness=-1
		
	\end{abstract}

	\section{INTRODUCTION}

	Technical systems are required to operate increasingly autonomously in uncertain environments. For ensuring safety and high performance, these systems need to be able to infer models from observed data online, such that they can quickly adapt to new situations. This is particularly important in applications such as the safe control of autonomous underwater vehicles~\cite{Sahoo2019}, unmanned aerial vehicles~\cite{Andersson2017} and wearable robots~\cite{Martinez-Hernandez2021}, where uncertainty arising from humans in the control loop and changing environments can prevent the derivation of accurate models prior to system operation.

	Gaussian process (GP) regression is a supervised machine learning method, which is commonly employed in highly nonlinear, safety-critical applications due to its high expressiveness and probabilistically bounded prediction errors~\cite{Rasmussen2006}. Even though it admits closed-form updates allowing online learning and thereby an iterative adaptation of inferred models, it exhibits a quadratic update complexity in the number of training samples. Therefore, it becomes too slow for processing streaming data generated during system operation in real-time, since controllers often run at sampling rates in the magnitude of \SI{e2}{Hz} to \SI{e3}{Hz} and consequently measurements quickly accumulate to large data sets, which render exact inference computationally intractable \cite{9-4_deisenroth2015distributed}. 
	In order to reduce the complexity of GPs, several approximations for online learning have been developed, which include inducing point methods \cite{3-4_huber2014recursive}, variational inference approaches \cite{Bui2017} 
	and finite feature approximations~\cite{4-4_gijsberts2013real}. While these approaches can yield computation times low enough for online learning in control, beneficial safety-relevant theoretical properties of exact GPs such as uniform error bounds \cite{0-5_lederer2019uniform} 
	do not directly extend to them, and thus, they cannot be used in safety-critical applications. In addition, those approaches exhibit a linear or even higher order polynomial memory complexity, which prohibits their application in resource-constrained technical systems such as drones, autonomous underwater vehicles or wearable robots with limited memory for storing data. In summary, there is a significant gap between the principle potential of GPs and their realistic application in safety-critical systems.\looseness=-1
	
	
		\begin{figure}
			\vspace{0.15cm}
			\centering
			\includegraphics[scale=0.38]{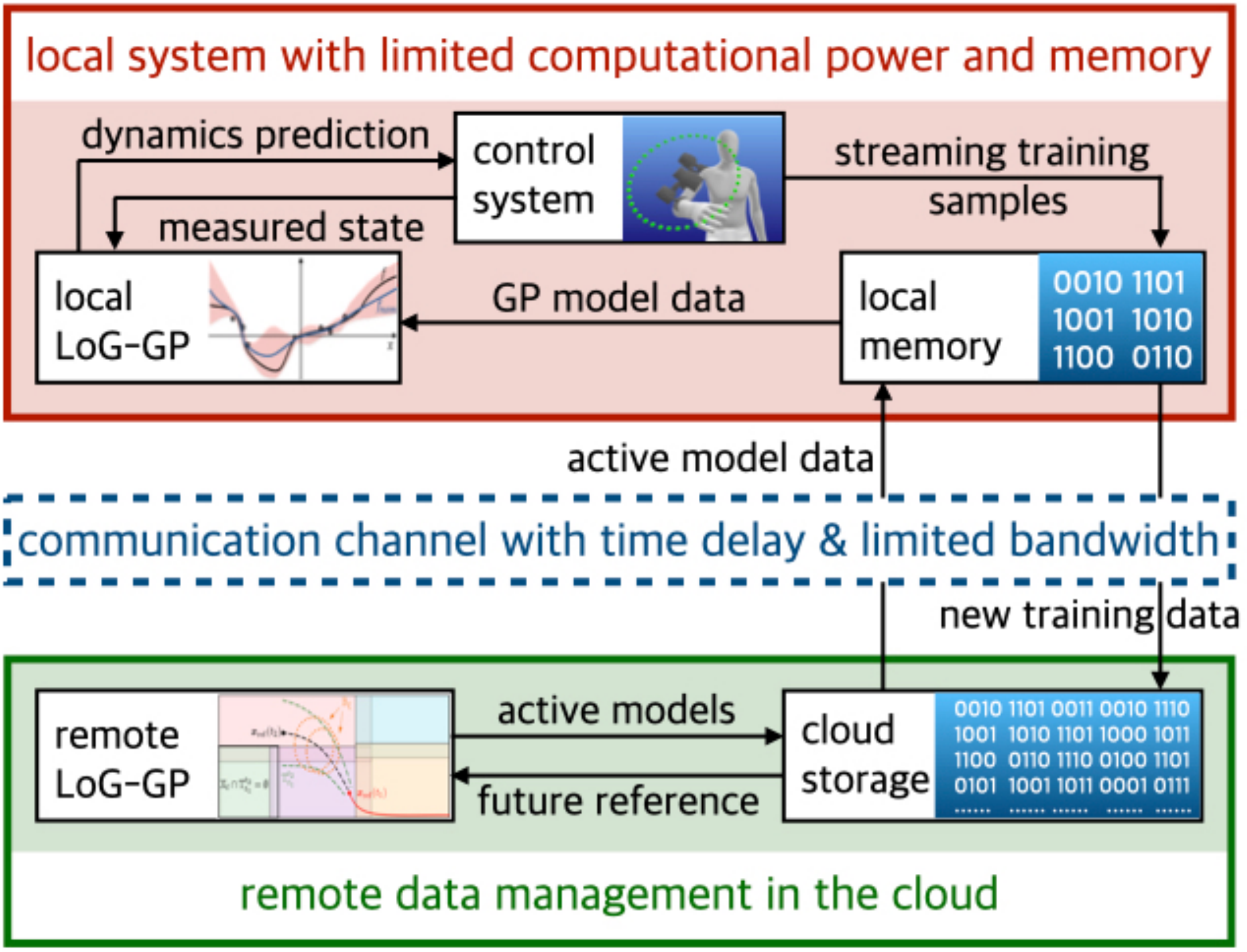}
			\vspace{-0.35cm}
			\caption{Overview of the proposed networked online learning architecture: The LoG-GP predicts the unknown dynamics, e.g., of a wearable robot, for a measured state. For computing these predictions, it can only access GP model data in the local memory. Measurements of the system are continuously stored in the local memory and regularly sent to the cloud, where necessary models for a future reference trajectory are determined using a sampling based approach and corresponding data is sent to the local memory.\looseness=-1
			}
			\label{fig:overview}
		\end{figure}
	
	This paper addresses the problem of online learning control for safety-critical systems with limited computational and memory resources. We exploit the fact that our envisioned applications are able to communicate with external infrastructure including clouds with potentially unlimited data storage. Realistic network restrictions such as time delays and limited bandwidth prevent the full externalization of the online model inference to the cloud. Therefore, we propose to learn models locally based on suitably small data sets and regularly exchanging relevant data between the local computation unit and the remote cloud as illustrated in \cref{fig:overview}. 
	Our approach employs locally growing random trees of GPs (LoG-GP)\footnote{Open-source software packages for several programming languages available at \url{https://gitlab.lrz.de/online-GPs/LoG-GPs}.} developed by the authors, which achieve logarithmically increasing update and prediction complexities while simultaneously inheriting error bounds from exact GPs~\cite{0-1_lederer2021gaussian}.
	In order to realize the data transfer without any loss in control performance, the local activity of Gaussian process models in LoG-GPs is exploited by selecting the most relevant local models for the immediate future. Based on a tracking error bound from GP-based learning control, we employ a sampling based approach to efficiently determine potentially relevant GP models with high probability,  such that only necessary data needs to remain in the local memory. 
	We ensure the timely availability of required data on the local system using an effective transmission scheme, which provides insight on fundamental trade-offs between the bandwidth, time delays, local memory and achievable tracking error. For demonstrating the effectiveness of the developed method, we exemplarily derive a novel tracking error bound for a feedback linearizing controller using a LoG-GP model, and evaluate the performance of the networked online learning approach in simulations of a robotic exoskeleton.

	The remainder of this paper is structured as follows: Section~\ref{sec:prob} formally describes the considered problem, followed by the proposed networked online learning method based on GPs in \cref{sec:main}. In \cref{sec:trackbound}, a tracking error bound for an online learning feedback linearizing control law is exemplarily derived, such that the effectiveness of the networked online learning approach can be demonstrated in \cref{sec:eval}, before the paper is concluded in \cref{sec:conc}.

	\section{PROBLEM DESCRIPTION}\label{sec:prob}

	Since accurate models for many systems such as autonomous underwater vehicles and wearable robots are often not available in practice, we consider the problem of inferring a dynamics model online from measurements generated during operation, such that the tracking performance of model-based control can be improved. 
	Formally, we model these systems with differential equations of the form\footnote{Notation: Lower/upper case bold symbols denote vectors/matrices, 
		$\mathbb{R}_+$/$\mathbb{R}_{0,+}$ all real positive/non-negative 
		numbers, $\bm{I}_n$ the $n\!\times\! n$ identity matrix, $\|\cdot\|$ the Euclidean norm, $|\mathbb{D}|$ the cardinality of a set $\mathbb{D}$, and  
		$\lceil\cdot\rceil$/$\lfloor\cdot\rfloor$ the ceil/floor operator.}\looseness=-1
	\begin{align}\label{equ:system}
		\dot{\bm{x}}=\bm{f}(\bm{x},\bm{u}),
	\end{align}
	where $\bm{x}\in\mathbb{X}\subset\mathbb{R}^{d_x}$ denotes the state, $\bm{u}\in\mathbb{U}\subset\mathbb{R}^{d_u}$ is the control input, and $\bm{f}:\mathbb{X}\times\mathbb{U}\rightarrow\mathbb{R}^{d_x}$ is the unknown dynamics function. We consider the task of tracking a bounded, continuously differentiable reference trajectory $\bm{x}_{\mathrm{ref}}:\mathbb{R}_{0,+}\rightarrow\mathbb{X}$ with the system state~$\boldsymbol{x}(t)$. For this purpose, we employ a model-based control law $\bm{\pi}_{\hat{f}}:\mathbb{X}\rightarrow\mathbb{U}$, where $\hat{\bm{f}}:\mathbb{X}\times\mathbb{U}\rightarrow\mathbb{R}^d$ is a model of the unknown function $\bm{f}(\cdot)$. 
	The tracking performance of such a control law typically depends strongly on the accuracy of the model $\hat{\bm{f}}(\cdot)$, such that we employ the following assumption on the model-based control law $\bm{\pi}_{\hat{f}}(\cdot)$, which is satisfied by many control techniques such as feedback linearization~\cite{0-6_umlauft2019feedback}, backstepping \cite{Capone2019} and adaptive control \cite{Gahlawat2020a}.
	
	\begin{assumption}\label{ass:errbound}
		The tracking error $\bm{e}(t)=\bm{x}-\bm{x}_{\mathrm{ref}}(t)$ is ultimately bounded with monotonously increasing ultimate bound $\vartheta:\mathbb{R}_{0,+}\rightarrow\mathbb{R}_{0,+}$, i.e., for every $c\in\mathbb{R}_+$, there exists a time $T=T(c,\vartheta)$, such that it holds that
		\begin{align}
		\|\bm{e}(0)\|\leq c \qquad \Rightarrow \qquad \|\bm{e}(t)\|\leq \vartheta(\kappa_t),~ \forall t\geq T,
		\end{align}
		where $\kappa_t=\max_{t'\in [0,t]}\|\bm{f}(\bm{x}(t'))-\hat{\bm{f}}(\bm{x}(t'))\|$.
	\end{assumption}
	For notational simplicity, we assume no knowledge of~$\bm{f}(\cdot)$ before system operation, but considering a prior model $\hat{\bm{f}}_{0}(\cdot)$ is straightforward \cite{0-6_umlauft2019feedback}. 
	In order to infer a model~$\hat{\bm{f}}(\cdot)$ online, we require periodical measurements of the system.
	\looseness=-1
	\begin{assumption}\label{ass:data}
		Data pairs $(\bm{x}_n,\bm{y}_n\!=\!f(\bm{x}_n,\bm{\pi}_{\hat{\bm{f}}}(\bm{x}))
		)\!+\!\bm{\epsilon}_n)$, where $\bm{\epsilon}_n\!\sim\!\mathcal{N}(0,\sigma_{\mathrm{on}}^2\bm{I}_{d_x})$ are i.i.d. Gaussian random variables with variance $\sigma_{\mathrm{on}}^2\!\in\!\mathbb{R}_+$, are sampled at 
		time instances $t^{(n)}\!=\!n\tau$ with sampling time $\tau\!\in\!\mathbb{R}_+$. The data is aggregated in a time-varying training set $\mathbb{D}_t\!=\!\{ (\bm{x}_n,\bm{y}_n) \}_{n=1}^{N(t)=\lfloor \frac{t}{\tau} \rfloor}$.
	\end{assumption}
	\cref{ass:data} admits training targets $\bm{y}$ perturbed by Gaussian noise, which is a frequently found assumption in literature, see, e.g.,~\cite{0-6_umlauft2019feedback, Capone2019, Gahlawat2020a}. It also requires noise-free state measurements for training, which however, is commonly assumed in many employed control schemes such as feedback linearization and sliding mode control~\cite{0-7_khalil2002nonlinear}.
	
	Since \cref{ass:data} ensures a continuous data stream, model updates of $\hat{\bm{f}}(\cdot)$ must be computed fast enough to avoid that data is generated at higher rates than it can be processed. Hence, the average update time $T_{\mathrm{up}}$ of $\hat{\bm{f}}(\cdot)$ must satisfy the computational constraint\looseness=-1
	\begin{align}\label{eq:compconst}
	    T_{\mathrm{up}}\leq \tau.
	\end{align}
	Additionally, the continuous stream of data 
	leads to a steadily growing size of the data set $\mathbb{D}_t$. 
	Therefore, the amount of generated data will eventually reach the memory limitations, which are unavoidable on all real-world systems. Formally, this can be modelled via the memory constraint
	\begin{align}\label{eq:memcon}
	|\mathbb{D}^{\mathrm{loc}}_t|\leq \bar{M},
	\end{align}
	where $\mathbb{D}^{\mathrm{loc}}_t$ denotes the data set stored in the memory of the technical system and $\bar{M}\!\in\!\mathbb{N}$ represents the memory limitations. Since this restriction can crucially limit the achievable control performance \cite{Lederer2020a}, we consider that data can be transferred to a cloud via a network connection, effectively extending the overall memory capacity. The available memory in the cloud is usually significantly larger than on the local system, such that we assume it to be infinite for simplicity. However, the data transfer between the cloud and the local system takes non-negligible time in practice due to effects such as network delays $T_d\!\in\!\mathbb{R}_+$ and finite bandwidth $B\!\in\!\mathbb{R}_+$. Therefore, data sent to the cloud cannot be immediately accessed by the local system, but the time $T_{\mathrm{access}}$ between requesting data $\mathbb{D}$ and using it has to satisfy the network constraint
	\begin{align}\label{eq:netcon}
	T_{\mathrm{access}} \geq \frac{|\mathbb{D}|}{B}+T_d.
	\end{align}
	Despite these restrictions, the model-based control law~$\bm{\pi}_{\hat{\bm{f}}}(\cdot)$ using the model~$\hat{\bm{f}}(\cdot)$ learned from the streaming data $\mathbb{D}_t$ should achieve a high tracking control performance. Therefore, we consider the problem of developing a networked online learning method for inferring a highly accurate model~$\hat{\bm{f}}(\cdot)$ of the unknown dynamics $\bm{f}(\cdot)$ under computational, memory and network constraints.

	\section{NETWORKED ONLINE LEARNING BASED ON GAUSSIAN PROCESSES}\label{sec:main}
	
	
	Since the time delay $T_d$ prevents externalizing the online learning%
	, we propose the networked online learning approach outlined in \cref{fig:overview}, which performs inference locally, but transfers unnecessary data to the cloud. The approach is based on GP regression~\cite{Rasmussen2006} due to its strong theoretical foundation as introduced in \cref{subsec:GP}. For enabling online learning with GPs, we employ LoG-GPs firstly proposed in our earlier work~\cite{0-1_lederer2021gaussian}, which inherit the probabilistic prediction error guarantees of exact GPs while having merely logarithmically increasing update and prediction complexities as outlined in \cref{subsec:LoG-GP}.
	In order to transmit data to the cloud without performance loss, we exploit the modular structure of LoG-GPs and determine the region, in which system states can potentially be in a given time interval, using a sampling-based approach in \cref{subsec:4.2}.  By developing a data transmission scheme 
	in \cref{subsec:4.3}, we ensure that necessary data is always locally available despite transmission bandwidth limitations and network delays. For notational simplicity, the proposed method is presented for scalar functions $f(\cdot)$, but can be employed for the vector-valued dynamics in \eqref{equ:system} by applying it to each dimension individually.\looseness=-1
	
	\subsection{Gaussian Process Regression}\label{subsec:GP}
	A Gaussian process is an infinite collection of random variables, any finite subset of which follows a joint Gaussian distribution \cite{Rasmussen2006}. The GP is usually denoted as $\mathcal{GP}(m(\cdot),k(\cdot,\cdot))$, where $m:\mathbb{R}^d\rightarrow\mathbb{R}$ is a prior mean incorporating a priori knowledge such as approximate models, and $k:\mathbb{R}^d\times\mathbb{R}^d\rightarrow\mathbb{R}_{0,+}$ is a covariance function reflecting information such as periodicity. Since we assume no prior knowledge, the prior mean~$m(\cdot)$ is set to $0$ in the sequel. Analogously, we employ the probably most common choice for the covariance function: the squared exponential kernel $k(\bm{x},\bm{x}')=\sigma_{f}^{2} \exp (-\sum_{i=1}^d (x_{i}-x_{i}^{\prime})^{2}/(2 l_i^{2}))$,
	where \mbox{$\sigma_{f}\in\mathbb{R}_+$} denotes the signal standard deviation, and $l_i\in\mathbb{R}_+$, \mbox{$i=1,\ldots,d$} are length scales \cite{Rasmussen2006}.
	
	Given a prior GP $\mathcal{GP}(0,k(\cdot,\cdot))$, regression is performed by conditioning on the training data $\mathbb{D}_t$ as introduced in \cref{ass:data}%
	. The resulting posterior distribution is again Gaussian with mean and variance given by\looseness=-1
	\begin{align}\label{eq:GP mean}
	\mu\left(\boldsymbol{x}\right)&=\bm{k}^T\left(\boldsymbol{x}\right)\left(\bm{K}+\sigma_{\mathrm{on}}^{2} \boldsymbol{I}\right)^{-1}\boldsymbol{y}\\
	\sigma^{2}\left(\boldsymbol{x}\right)&=k\left(\boldsymbol{x}, \boldsymbol{x}\right)-\bm{k}^T\left(\boldsymbol{x}\right)\left(\bm{K}+\sigma_{\mathrm{on}}^{2} \boldsymbol{I}\right)^{-1}\bm{k}\left(\boldsymbol{x}\right),\label{eq:GP variance}
	\end{align}
	where the elements of $\bm{K}\!\in\!\mathbb{R}^{N\!\times\! N}$ and $\bm{k}(\bm{x})\!\in\!\mathbb{R}^N$ are defined through $K_{i,j}\!=\!k(\bm{x}_i\!,\bm{x}_j)$ and $k_i(\bm{x})\!=\!k(\bm{x},\bm{x}_i)$, respectively, and we concatenate training targets $\bm{y}\!=\![y_1\ \cdots\ y_N]^T\!$.\looseness=-1

	\subsection{Locally Growing Random Tree of Gaussian Processes} \label{subsec:LoG-GP}
	
	Since the update complexity of Gaussian process regression scales quadratically with the number of training samples, we employ the recently proposed approach of locally growing random trees of GPs \cite{0-1_lederer2021gaussian}, which preserves beneficial properties of exact GP inference such as the existence of uniform prediction error bounds. 
	LoG-GPs rely on the idea of iteratively constructing a tree, whose leaf nodes contain locally active GP models. In detail, the construction starts with a single GP model, which is updated with incoming streaming data until a prescribed threshold of training samples $\bar{N}$ is reached. When the GP model contains $\bar{N}$ training samples in its data set $\mathbb{D}_{0}$, the data set is split into $2$ subsets $\mathbb{D}_i$, $i=1,2$, by assigning data in $\mathbb{D}_0$ to a subset $\mathbb{D}_i$ via sampling from a Lipschitz continuous probability function $p^0:\mathbb{R}^d\rightarrow [0,1]$. Thereby, a tree with $2$ leaf nodes is generated, which contain all the data, such that individual GP models can be efficiently computed using \eqref{eq:GP mean} and \eqref{eq:GP variance}. New streaming data obtained after the splitting can be assigned to the subsets $\mathbb{D}_i$ by sampling from $p^0(\cdot)$ again until either of the subsets $\mathbb{D}_i$ reaches the capacity limit $\bar{N}$. Then, a new probability function $p^i:\mathbb{R}^d\rightarrow [0,1]$ is defined to distribute the data to new subsets, thereby extending the tree of GPs by a new layer. By repeating this procedure every time a subset $\mathbb{D}_i$ reaches $\bar{N}$ training samples, a tree of GP models is iteratively constructed with a computational complexity of $\mathcal{O}_p(\log(N))$ allowing updates with rates up to $1\si{kHz}$ \cite{0-1_lederer2021gaussian}, which is fast enough to satisfy the computational constraint \eqref{eq:compconst} in many systems.\looseness=-1
	
	For computing predictions with LoG-GPs, we simply multiply the probabilities $p^i(\bm{x})$ along a path to a leaf node~$l$ to obtain the weight $\omega_l(\bm{x})$. Then, a generalized product of experts aggregation scheme \cite{9-4_deisenroth2015distributed} can be employed to obtain the approximate GP prediction
	\begin{align}\label{eq:log-gp mean}
	\!\tilde{\mu}(\bm{x})\!=\!\sum\limits_{l\in\mathbb{L}}\!\frac{\omega_l(\bm{x})\tilde{\sigma}^2(\bm{x})}{\sigma_l^2(\bm{x})}\mu_l(\bm{x}),&&
	\tilde{\sigma}^{-2}(\bm{x})\!=\!\sum\limits_{l\in\mathbb{L}}\frac{\omega_l(\bm{x})}{\sigma_l^2(\bm{x})},\!
	\end{align}
	where $\mathbb{L}$ denotes the set of leaf nodes of the tree of GP models. By defining the probability functions such that only a single child node has a positive probability in most of the input domain $\mathbb{R}^d$, most of the weights $\omega_l(\bm{x})$ become~$0$. Since the definition of the aggregated mean $\tilde{\mu}(\cdot)$ implies that the local GP predictions $\mu_l(\bm{x})$ and $\sigma_l^2(\bm{x})$ must only be computed if $\omega_l(\bm{x})>0$, models with $\omega_l(\bm{x})=0$ can be considered locally inactive at $\bm{x}$ and therefore, aggregated predictions can be efficiently computed in $\mathcal{O}_p(\log^2(N))$ complexity. 
	Moreover, this construction of the aggregated prediction~$\tilde{\mu}(\cdot)$ ensures that uniform error bounds are directly inherited from exact GP regression. 
	\begin{lemma} [\cite{0-1_lederer2021gaussian}] \label{the:Regression Error Bound}
		Assume the function $f:\mathbb{R}^d\rightarrow\mathbb{R}$ is a sample from a Gaussian process $\mathcal{GP}(0,k(\cdot,\cdot))$ with a $L_k$-Lipschitz kernel $k:\mathbb{R}^d\times\mathbb{R}^d\rightarrow\mathbb{R}$. Then, the aggregated mean function \eqref{eq:log-gp mean} of a LoG-GP trained with data satisfying \cref{ass:data} guarantees a probabilistically, uniformly bounded prediction error on a compact domain $\Omega\subset\mathbb{R}^d$, i.e., for $\delta\in(0,1)$ and $\rho\in\mathbb{R}_+$, we have
		\begin{equation} \label{equ:AA}
		P(|f(\boldsymbol{x})-\tilde{\mu}(\boldsymbol{x})| \leq \eta( \boldsymbol{x}), \forall \boldsymbol{x} \in \Omega) \geq 1-\delta,\end{equation} 
		where\allowdisplaybreaks
		\begin{align} \label{eq:BB}
		\eta( \boldsymbol{x})&=\!\sqrt{\beta(\delta,\rho)} \sum_{l \in \mathbb{L}} \frac{\omega_{l}(\bm{x})\tilde{\sigma}^2(\bm{x})}{\sigma_l(\bm{x})} +\gamma(\rho)\\
		\!\beta(\delta,\rho)&=\!2 \log\! \Big(\!d^{\frac{d}{2}} |\mathbb{L}|\! \max _{\boldsymbol{x}, \boldsymbol{x}^{\prime} \!\in \mathbb{R}^d}\!\! \left\|\boldsymbol{x}\!-\!\boldsymbol{x}^{\prime}\right\|_{\infty}^{d}\!\!\Big)\!-\!2\log\!\left(\delta 2^d\rho^d\right)\!\!\\
		\gamma(\rho)&=\!\sum_{l \in \mathbb{L}} \frac{\omega_{l}\tilde{\sigma}^2(\bm{x})}{\sigma_l^2(\bm{x})}\!\Big(L_{\mu_{l}} \rho+\sqrt{\beta(\rho)} L_{\sigma_{l}} \tau\Big)+L_{f} \rho,\!
		\end{align} 
		and $L_f$, $L_{\mu_l}$, $L_{\sigma_l}$ are Lipschitz constants of $f(\cdot)$, $\mu_l(\cdot)$, $\sigma_l(\cdot)$.\looseness=-1
	\end{lemma}
	
	This result relies on a well-calibrated prior GP
	, which is a rather unrestrictive assumption in practice \cite{0-5_lederer2019uniform}. 
	Therefore, LoG-GPs provide strong theoretical guarantees for their prediction accuracy as required in safety critical applications. 
	
	
	
	\setlength{\textfloatsep}{7pt}
	\setlength{\floatsep}{0pt}
	

		\subsection{Sampling-Based Identification of Active Models} \label{subsec:4.2}
		
		Since \cref{the:Regression Error Bound} ensures bounded prediction errors when using \eqref{eq:log-gp mean} to learn a model $\hat{\bm{f}}(\cdot)$ of the unknown dynamics~$\bm{f}(\cdot)$, we can determine the system states $\bm{x}$ which can be potentially reached within a fixed time interval using the tracking error bound $\vartheta(\kappa_t)$ introduced in \cref{ass:errbound}. Therefore, we can obtain the models, which need to be available in the local memory, by finding all individual GP models which are active for states $\bm{x}$ in the potentially reachable set. \looseness=-1
		
		In detail, this set $\mathbb{A}$ of potentially active models during a time window $\mathbb{W}\!=\![t_1,t_2]$, $t_1,t_2\!\in\!\mathbb{R}$, $t_2\!>\!t_1$, is defined through the intersections between active regions $\mathbb{X}_l\!=\!\{\bm{x}\!: \omega_l(\bm{x})\!>\!0\}$ of local models $l\!\in\!\mathbb{L}$ and the tube $\mathbb{T}_{t_1}^{t_2}\!=\!\{\bm{x}\!\in\!\mathbb{R}^d\!: \exists t\!\in\!\mathbb{W}, \bm{x}\!\in\!\mathbb{B}_{\vartheta(\kappa_t)}  \}$ based on balls $\mathbb{B}_{\vartheta(\kappa_t)}\!=\!\{\bm{x}\!\in\!\mathbb{R}^d\!: \|\bm{x}\!-\!\bm{x}_{\mathrm{ref}}(t)\|\!\leq\! \vartheta(\kappa_t)\}$ with radius given by \cref{ass:errbound}, i.e., $\mathbb{A}\!=\!\{l\!\in\!\mathbb{L}\!:  \mathbb{X}_l\cap\mathbb{T}_{t_1}^{t_2}\!\neq\! \emptyset \}$
		\begin{figure}[t] 
			\definecolor{FigureRed}{RGB}{255,0,0}
			\definecolor{FigureOrange}{RGB}{245,123,0}
			\definecolor{FigureGreen}{RGB}{56,142,61}
			
			\centering
			\begin{tikzpicture}[font=\sffamily]
			\node at (0,0) {\includegraphics[width=0.445\textwidth]{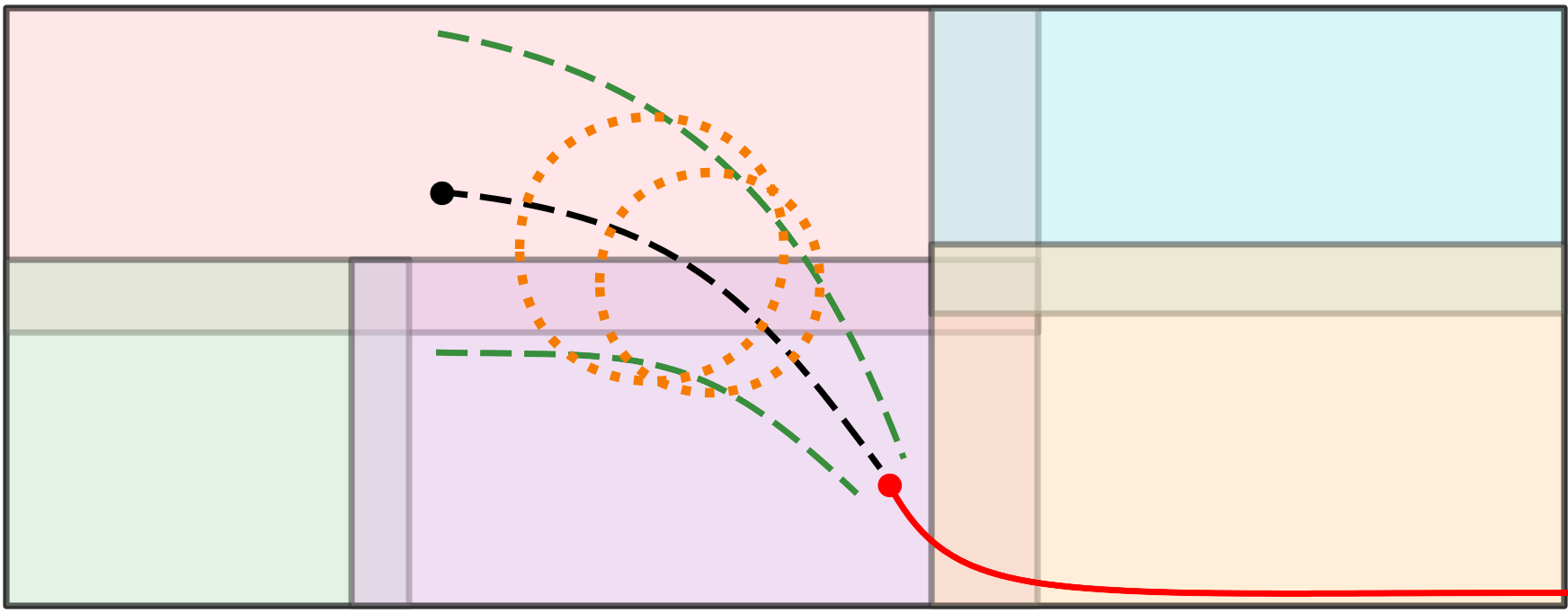}};
			\node[FigureGreen] at (-1.55,-0.55) {\footnotesize $\mathbb{T}_{t_1}^{t_2}$};
			\node at (-2.3,0.6) {\footnotesize $\bm{x}_{\mathrm{ref}}(t_2)$};
			\node[FigureRed] at (1.3,-0.9) {\footnotesize $\bm{x}_{\mathrm{ref}}(t_1)$};
			\node[FigureOrange] at (0.45,1.3) {\footnotesize $\mathbb{B}_{\xi}$};
			\draw[-Latex, FigureOrange] (0.2,1.2) -- (-0.15,0.9);
			\draw[-Latex, FigureOrange] (0.3,1.1) -- (0.2,0.55);
			
			\draw[thick] (0.75,0.33) rectangle (3.95,1.51);
			\node at (2.6,0.92) {\footnotesize $\mathbb{X}_l\cap\mathbb{T}_{t_1}^{t_2}= \emptyset$};
			\end{tikzpicture}
			\vspace{-0.4cm}
			\caption{A local model $l\in\mathbb{L}$ is inactive if its active region $\mathbb{X}_l$ does not intersect with the tube $\mathbb{T}_{t_1}^{t_2}$ induced by the tracking error bound $\vartheta(\kappa_t)$ as illustrated for the region in the top right. The set of active models $\hat{\mathbb{A}}$ is found by over-approximating the tube $\mathbb{T}_{t_1}^{t_2}$ with balls $\mathbb{B}_{\xi}$, from which random samples $\bm{x}^{(i)}$ are drawn to determine the active models $\mathbb{A}_{\bm{x}^{(i)}}$ at these states.\looseness=-1
			}
			\label{fig:explain}
		\end{figure}
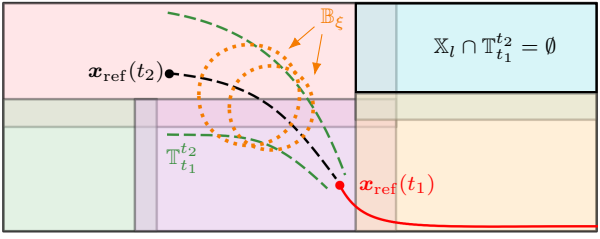
		as illustrated in \cref{fig:explain}. Since the computation of the intersections $\mathbb{X}_l\cap \mathbb{T}_{t_1}^{t_2}$ requires an explicit representation of the active regions $\mathbb{X}_l$ of local models $l\in\mathbb{L}$, which is not provided by LoG-GPs, the definition of $\mathbb{A}$ cannot be directly used in practice. We follow a different idea exploiting the implicit representation of the active regions $\mathbb{X}_l$ via the weights $\omega_l(\cdot)$, which allows to directly compute the set of active models $\mathbb{A}_{\bm{x}}=\{l\!\in\!\mathbb{L}\!: \omega_l(\bm{x})\!>\!0 \}$ for a given state $\bm{x}$. Therefore, we can alternatively represent the set of potentially active models during the time window $\mathbb{W}$ via $\mathbb{A}=\bigcup\nolimits_{t\in\mathbb{W}}\bigcup\nolimits_{\bm{x}\in \mathbb{B}_{\vartheta(\kappa_t)}}\mathbb{A}_{\bm{x}}$.
		By approximating the unions over uncountable sets via discretization and random sampling as outlined in Algorithm~\ref{alg:decrease}, we can over-approximate the set $\mathbb{A}$ via $\hat{\mathbb{A}}$ and obtain\looseness=-1
		\begin{align}\label{eq:muhat}
		\!\hat{\mu}(\bm{x})\!=\!\sum\limits_{l\in\hat{\mathbb{A}}}\!\frac{\omega_l(\bm{x})\hat{\sigma}^2(\bm{x})}{\sigma_l^2(\bm{x})}\mu_l(\bm{x}),
		&&\hat{\sigma}^{-2}(\bm{x})\!=\!\sum\limits_{l\in\hat{\mathbb{A}}}\frac{\omega_l(\bm{x})}{\sigma_l^2(\bm{x})}.\!
		\end{align}
		If sufficiently many samples are used, this approximation yields identical predictions as shown in the following result.\looseness=-1
		
		\begin{algorithm}[t]
			\small
			\SetInd{0.6em}{0.5em}
			\DontPrintSemicolon
			\SetKwFunction{FSOP}{ActiveModels}
			\SetKwProg{Fn}{Function}{:}{}
			\Fn{\FSOP{$N_s$, $t_1$, $t_2$, $\Delta t$}}{
				$\hat{\mathbb{A}}\gets\emptyset$\;
				compute $\xi$ using \eqref{eq:xi}\;
				\For{$j=0:\lceil\frac{t_2-t_1}{\Delta t} \rceil$}{
					\For{$i=1:N_s$}{
						Determine active models $\mathbb{A}_{\bm{x}^{(i)}}$ for input $\bm{x}^{(i)}\!\!\sim\!\mathcal{U}(\mathcal{B}_{\xi})$\;
						$\hat{\mathbb{A}}\gets\hat{\mathbb{A}}\cup\mathbb{A}_{\bm{x}^{(i)}}$\;
					}
				}
			}
			\KwRet $\hat{\mathbb{A}}$
			\caption{Determining Active Models}
			\label{alg:decrease}
		\end{algorithm}
		
		\begin{theorem}\label{th:sample_active}
			Consider a dynamical system~\eqref{equ:system} and assume Assumptions~\ref{ass:errbound} and~\ref{ass:data} hold. Choose $\|\bm{e}(t_1)\|\leq \vartheta(\kappa_{t_1})$ and\looseness=-1
			\begin{align}\label{eq:xi}
			\xi = 2\zeta + L_{\bm{x}_{\mathrm{ref}}} \frac{\Delta t}{2}+\vartheta(\kappa_{t_1+j\Delta t})
			\end{align}
			for constants $\zeta,\Delta t\in\mathbb{R}_+$. Then, with probability of at least 
			\begin{align}\label{eq:activeModelProb}
			\!1\!-\! \|\mathbb{L}\|\!\left\lceil\!\frac{t_2\!-\!t_1}{\Delta t} \!\right\rceil\!\! \left(\! 1\!-\!\frac{\min\{r_{\min}^{d_x},\zeta^{d_x}\}}{\xi^{d_x}}\right)^{\!\!N_s}\!,\!
			\end{align}
			where $r_{\min}$ denotes the radius of the largest ball contained in the smallest active region of a leaf node $l\in\mathbb{L}$,
			the predictions $\hat{\mu}(\bm{x}(t))$ and $\tilde{\mu}(\bm{x}(t))$ are identical for all $t\in\mathbb{W}$.
		\end{theorem}
		\begin{proof}
			See \cref{app:A}.
		\end{proof}
		Since this theorem ensures that \eqref{eq:log-gp mean} and \eqref{eq:muhat} are identical with probability greater than \eqref{eq:activeModelProb}, it ensures that using $\hat{\mu}(\cdot)$ as model in a control law $\bm{\pi}_{\hat{\bm{f}}}(\cdot)$ yields no reduction in control performance with high probability. Therefore, it allows us to determine irrelevant data for a time interval $\mathbb{W}$, which we exploit in the following section for transmitting data to the cloud, thereby reducing the local memory occupation.

	\subsection{Transmission Scheme}\label{subsec:4.3}
	
	Due to the non-negligible time required for a data transfer, the transmission to and from the cloud must be carefully scheduled in order to ensure that 
	the necessary data is always available locally. For simplicity, we consider that data is transmitted at regularly spaced time instances $j\Delta T$, $j=\mathbb{N}$, such that each time interval $\mathbb{W}_j=[(j-1)\Delta T, j\Delta T]$ has a length of $\Delta T\in\mathbb{R}_+$. During each time interval $\mathbb{W}_j$, we propose the transmission scheme illustrated in \cref{fig:trans_scheme}, where the idea is that the memory is divided into two parts. During each interval $\mathbb{W}_j$, half of the memory is used for updating the local data set with data from the cloud, while the other half contains the data set $\mathbb{D}_j$ necessary for computing the mean predictions~$\hat{\mu}(\cdot)$ during time interval $\mathbb{W}_j$ according to the potentially active models $\hat{\mathbb{A}}_{j}$. For updating the local memory, the data set $\mathbb{D}_{j-1}$ from the previous interval $\mathbb{W}_{j-1}$, which contains newly measured training samples as well as data from the cloud, is sent to the cloud. Once this transmission has been completed, the cloud contains the complete data set $\mathbb{D}_{(j-1)\Delta T}$ obtained until time $(j-1)\Delta T$, such that Algorithm~\ref{alg:decrease} can be employed to determine the possibly active models $\hat{\mathbb{A}}_{j+1}$ for the next time interval $\mathbb{W}_{j+1}$ in the cloud. The corresponding data set $\mathbb{D}_{j+1}$ is sent to the local memory, such that it is available for $t\geq (j+1)\Delta T$. \looseness=-1

		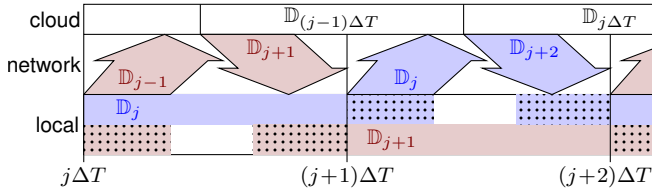
\begin{figure}
		\centering
		\begin{tikzpicture}[font=\sffamily]
		\clip (1.0,-0.55) rectangle (9.55,2.0);
		\draw (2,-0.1) rectangle (9.65,0.7);
		
		\draw (2,1.5) rectangle (9.65,1.9);
		
		\fill[red!20] (2,-0.1) rectangle (3.15,0.3);
		\fill[pattern=dots] (2,-0.1) rectangle (3.15,0.3);
		\fill[blue!20] (2,0.3) rectangle (6.65,0.7);
		\fill[pattern=dots] (5.5,0.3) rectangle (6.65,0.7);
		
		\fill[red!20] (4.25,-0.1) rectangle (9.65,0.3);
		\fill[pattern=dots] (4.25,-0.1) rectangle (5.5,0.3);
		\fill[pattern=dots] (9,-0.1) rectangle (9.65,0.3);
		
		\fill[blue!20] (7.75,0.3) rectangle (9.65,0.7);
		\fill[pattern=dots] (7.75,0.3) rectangle (9.0,0.7);
		
		\draw[fill=red!20] (2,0.7) -- (2.4,1.15) -- (2.15,1.15) -- (3.075,1.5) -- (3.85,1.15) -- (3.65,1.15) -- (3.15,0.7) -- (2,0.7);
		\draw[fill=red!20] (3.55,1.5) -- (4.05,1.05) -- (3.8,1.05) -- (4.725,0.7) -- (5.5,1.05) -- (5.3,1.05) -- (4.8,1.5) -- (3.55,1.5);
		
		\draw[fill=blue!20] (5.5,0.7) -- (5.9,1.15) -- (5.65,1.15) -- (6.575,1.5) -- (7.35,1.15) -- (7.15,1.15) -- (6.65,0.7) -- (5.5,0.7);
		\draw[fill=blue!20] (7.05,1.5) -- (7.55,1.05) -- (7.3,1.05) -- (8.225,0.7) -- (9,1.05) -- (8.8,1.05) -- (8.3,1.5) -- (7.05,1.5);
		
		\draw[fill=red!20] (9,0.7) -- (9.4,1.15) -- (9.15,1.15) -- (10.075,1.5) -- (10.85,1.15) -- (10.65,1.15) -- (10.15,0.7) -- (9,0.7);
		
		\draw[] (2,-0.2) -- (2,1.5);
		\node at (2,-0.4) {\footnotesize $j\Delta T$};
		\draw[] (5.5,-0.2) -- (5.5,1.5);
		\node at (5.5,-0.4) {\footnotesize $(j\!+\!1)\Delta T$};
		\draw[] (9,-0.2) -- (9,1.5);
		\node at (8.93,-0.4) {\footnotesize $(j\!+\!2)\Delta T$};
		
		\draw[] (3.55,1.5) -- (3.55,1.9);
		\draw[] (7.05,1.5) -- (7.05,1.9);
		
		\node at (1.65,0.3) {\footnotesize local};
		\node at (1.47,1.1) {\footnotesize network};
		\node at (1.62,1.7) {\footnotesize cloud};

		\node at (2.6,0.5) {\footnotesize \textcolor{blue}{$\mathbb{D}_{j}$}};
		\node at (6.1,0.1) {\footnotesize \textcolor{red}{$\mathbb{D}_{j+1}$}};

		\node at (2.8,0.9) {\footnotesize \textcolor{red}{$\mathbb{D}_{j-1}$}};
		\node at (4.5,1.3) {\footnotesize \textcolor{red}{$\mathbb{D}_{j+1}$}};
		
		\node at (6.3,0.9) {\footnotesize \textcolor{blue}{$\mathbb{D}_{j}$}};
		\node at (8,1.3) {\footnotesize \textcolor{blue}{$\mathbb{D}_{j+2}$}};
		
		\node at (5.3,1.7) {\footnotesize $\mathbb{D}_{(j-1)\Delta T}$};
		\node at (9,1.7) {\footnotesize $\mathbb{D}_{j\Delta T}$};

		\end{tikzpicture}
		\vspace{-0.6cm}
		\caption{During each interval $\mathbb{W}_j=[j\Delta T,(j+1)\Delta T]$, the previously necessary data $\mathbb{D}_{j-1}$ is sent to the cloud and the data $\mathbb{D}_{j+1}$ for the next interval $\mathbb{W}_{j+1}$ is fetched. While these data sets occupy memory during the interval $\mathbb{W}_j$, parts of $\mathbb{D}_{j-1}$ and $\mathbb{D}_{j+1}$ are in transmission and not available on the local system. Therefore, these data sets cannot be used for prediction, which is highlighted through the dotted pattern. The data in the cloud is updated with incoming transmissions, such that it contains the complete data set $\mathbb{D}_{(j-1)\Delta T}$ up to the end of previous interval $j-1$.\looseness=-1
		}
		\label{fig:trans_scheme}
	\end{figure}

	It is straightforward to see that this transmission scheme can ensure the satisfaction of the network constraint \eqref{eq:netcon} for  a fixed data set $\mathbb{D}_j$, if $T_{\mathrm{access}}= \Delta T/2$ is sufficiently large. However, due to the online generation of data during system operation, it generally cannot be ensured that the data sets $\mathbb{D}_j$ have a bounded size, such that the fixed time $\Delta T$ might eventually not be sufficient to finish the transmission within the time interval $\mathbb{W}_j$. Therefore, the real-time learning with data generated online during system operation has to be stopped eventually at some interval $\mathbb{W}_{\iota}$, $\iota\in\mathbb{N}$ in order to upper bound the size of all sets $\mathbb{D}_j$. This leads to the data transfer scheme outlined in Algorithm~\ref{alg:cloud} for the cloud and in Algorithm~\ref{alg:local} for the local system, for which it is straightforward to prove the satisfaction of the network constraint \eqref{eq:netcon} as shown in the following result.
	
		\begin{lemma}\label{lem:netcon}
		Choose $\Delta T\geq \frac{\bar{M}}{B}+2T_d$ and $\iota\in\mathbb{N}$ such that the memory constraint \eqref{eq:memcon} is satisfied. Then, Algorithms~\ref{alg:cloud} and~\ref{alg:local} ensure the satisfaction of the network constraint \eqref{eq:netcon}.
	\end{lemma}
	\begin{proof}
		See \cref{app:B}.
	\end{proof}
	
	In order to apply this lemma in a real-world system, it remains to develop an approach for enforcing the memory constraint \eqref{eq:memcon} by choosing a suitable value of $\iota$. In practice, this value can be selected online using heuristics such that learning can be stopped, e.g., when the number of active models exceeds a threshold. Moreover, when the reference $\bm{x}_{\mathrm{ref}}$ is periodic, we can determine $\iota$ based on the data sets from previous periods, as shown in the following theorem.\looseness=-1
	\begin{theorem}\label{th:transscheme}
		Assume the reference trajectory is periodic with period $T_p=q\Delta T$ for $\Delta T\geq \frac{\bar{M}}{B}+2T_d$ and $q\in\mathbb{N}$. Let
		\begin{align}\label{eq:iota}
		\iota&=q+\min_{|\mathbb{D}_{j}|> \frac{\bar{M}}{2}-\bar{m}} j\\
		\bar{m}&=\max\limits_{j\in\mathbb{N}}|\mathbb{D}_{(j+q)}|-|\hat{\mathbb{D}}_{j}|\leq \left\lceil\frac{ T_p}{\tau}\right\rceil.
		\end{align}
		Then, Algorithms~\ref{alg:cloud} and~\ref{alg:local} ensure the satisfaction of the memory constraint \eqref{eq:memcon} and network constraint \eqref{eq:netcon}.
	\end{theorem}
	\begin{proof}
		See \cref{app:C}.
	\end{proof}
	
	This theorem allows to determine online when to stop adding new training samples to the LoG-GP by checking if $|\mathbb{D}_{j}|> \frac{\bar{M}}{2}-\bar{m}$, which can be performed with low complexity and can be directly implemented. Moreover, it provides valuable insight into the interrelations between achievable tracking accuracy, memory constraint $\bar{M}$, time delay $T_d$ and limited bandwidth $B$. In order to see this, note that the data set size $|\mathbb{D}_j|$ usually grows almost linearly with the interval length $\Delta T$. Since an increase in bandwidth $B$ admits smaller $\Delta T$, learning can continue up to higher values of $\iota$ in general. Therefore, a higher data density can be achieved, which in turn yields a lower GP variance~\cite{Lederer2019a} guaranteeing a smaller tracking error. In contrast, an increase in local memory $\bar{M}$ admits larger data set sizes $|\mathbb{D}_j|$, but in turn requires longer intervals $\Delta T$, such that the achievable data density and consequently the tracking accuracy are barely affected. Finally, a reduction of the delay $T_d$ allows smaller values of $\Delta T$ and thereby also leads to an improvement in achievable control performance. Therefore, available bandwidth $B$ for data transmission and time delay $T_d$ are crucial for the achievable tracking accuracy when using the networked online learning control law, while finite local memory $\bar{M}$ only has secondary relevance to enable implementation of the transmission scheme using Algorithm~\ref{alg:cloud} and~\ref{alg:local}. This insight can be beneficially used for the design of autonomous systems in practice, since it allows a reduction of local memory when sufficient bandwidth for data transmission is available. \looseness=-1

	\begin{algorithm}[t]
		\small
		\SetInd{0.7em}{0.6em}
		\DontPrintSemicolon
		\SetKwFunction{FSOP}{ActiveModels}
		\SetKwFunction{transmit}{Transmit}
		\SetKwFunction{receive}{Receive}
		\SetKwFunction{update}{Update}
		\SetKwFunction{Loop}{UpdateLoop}
		\SetKwProg{Fn}{Function}{:}{}
		\Fn{\Loop{$\Delta T$, $\tau$, $\Delta t$}}{
			\For{$n=1,\ldots,\infty$}{
				\If{$n\tau \geq j\Delta T$}{
					$j\gets j+1$\;
					$\mathbb{D}_{j-1}\gets$\receive()\;
					$\hat{\mathbb{A}}_{j+1}\!\gets$\FSOP$\!\!(N_s, j\Delta T,(j\!+\!1)\Delta T,\Delta t)$\;
					\transmit($\mathbb{D}_{j+1}$)\;
				}
			}
		}
		\caption{Data Transfer Scheme: Cloud}
		\label{alg:cloud}
	\end{algorithm}
	
	\begin{algorithm}[t]
		\small
		\SetInd{0.7em}{0.6em}
		\DontPrintSemicolon
		\SetKwFunction{FSOP}{getActiveModels}
		\SetKwFunction{transmit}{Transmit}
		\SetKwFunction{receive}{Receive}
		\SetKwFunction{update}{Update}
		\SetKwFunction{delete}{Delete}
		\SetKwFunction{Loop}{UpdateLoop}
		\SetKwProg{Fn}{Function}{:}{}
		\Fn{\Loop{$\Delta T$, $\iota$, $\tau$}}{
			\For{$n=1,\ldots,\infty$}{
				\If{$n\tau\leq \iota \Delta T$}{
					$\mathbb{D}_t^{\mathrm{loc}}\gets \mathbb{D}_t^{\mathrm{loc}}\cup (\bm{x}^{(n)},y^{(n)})$\;
				}
				\If{$n\tau \geq j\Delta T$}{
					
					$j\gets j+1$\;
					\transmit($\mathbb{D}_{j-1}$)\;
					\delete($\mathbb{D}_{j-1}$)\;
					$\mathbb{D}_{j+1}\gets$\receive()\;
				}
			}
		}
		\caption{Data Transfer Scheme: Local System}
		\label{alg:local}
	\end{algorithm}

	\section{EFFICIENT TRACKING ERROR BOUNDS}\label{sec:trackbound}
	
	In order to demonstrate the applicability of the proposed networked online learning approach, we exemplarily derive a tracking error bound $\vartheta(\kappa_t)$ for a feedback linearizing control law, which can be applied to a wide range of practically relevant systems such as robotic manipulators, unmanned aerial and autonomous underwater vehicles. For the derivation of $\vartheta(\kappa_t)$ we employ Lyapunov stability theory, such that computing $\vartheta(\kappa_t)$ effectively reduces to determining the regions of the state space with decreasing Lyapunov function along system trajectories. Due to the prediction error bound $\eta(\cdot)$, this decrease condition can be efficiently decoupled for feedback linearizing control laws as illustrated in \cref{fig:trackbound}. Thereby, we obtain a straightforwardly implementable tracking error bound, which can be directly used in Algorithms~\ref{alg:cloud} and~\ref{alg:local}.\looseness=-1
	
	\begin{figure}
		\centering
			\begin{tikzpicture}[font=\sffamily,scale=0.97]
					
					\fill[red!20] (4.0,1.47) -- (0.6,1.8) -- (0.6,-0.1) -- (4.0,1.47);
			
					\draw[blue, {latex}-{latex}] (0.6,0.85)--(0.6,1.4);
					\draw[cyan, {latex}-{latex}] (0.6,1.4)--(0.6,1.8);
					
					\draw[blue, {latex}-{latex}] (0.6,0.85)--(0.6,0.3);
					\draw[cyan, {latex}-{latex}] (0.6,0.3)--(0.6,-0.1);
					
					\begin{scope}
					\clip(-0.45cm,-0.3cm) rectangle (6.5cm,2.4cm);
					\foreach \i in {1.4}{
						\draw (0,0) ellipse (\i*4cm and \i*1.5 cm);}
					\end{scope}
					
					\node[anchor = north] at (6.0,0.25) {\footnotesize $V(\bm{e})$};
					\draw[very thick]  (4.0,1.47) node[cross=0.1cm] (xk) {};
					\node[anchor = north] at (xk) {\footnotesize $\bm{e}$};
					
					
					\draw[very thick]  (0,0) node[cross=0.1cm] (O) {};
					\node[anchor = north] at (O) {\footnotesize $\bm{e}=0$};
					
					\draw[-{latex}] (xk) -- (0.6,0.85);
					
					\node at (1.8,1.3) {\footnotesize  $\dot{\bm{e}}\!=\!\bm{A}\bm{e}$};
					\node[blue] at (-0.6,0.6) {\footnotesize  $\begin{bmatrix}0\vspace{-0.2cm}\\\vdots\vspace{-0.05cm}\\1 \end{bmatrix}\!\!(L_f\!+\!L_{\tilde{\mu}})\|\bm{e}\|$};
					\node[cyan] at (-0.2,1.6) {\footnotesize  $\begin{bmatrix}0\vspace{-0.2cm}\\\vdots\vspace{-0.05cm}\\1 \end{bmatrix}\!\eta(\bm{x}_{\mathrm{ref}})$};
					
					\draw[-{latex}, dashed, red] (xk) -- (0.6,0.1);
					
					
					\node at (2.6,0.55) [red] {\footnotesize  unknown};
					\node at (2.6,0.2) [red] {\footnotesize  dynamics};
			
			\end{tikzpicture}
		\vspace{-0.4cm}
		\caption{Ultimate boundedness is analyzed using Lyapunov theory with Lyapunov function $V(\bm{e})=\bm{e}^T\bm{P}\bm{e}$. For exact feedback linearization, the closed-loop dynamics are linear with $\dot{\bm{e}}=\bm{A}\bm{e}$ and ensure a decreasing Lyapunov function along system trajectories. Due to model errors, which can be expressed through the model error bound $\eta(\bm{x}_{\mathrm{ref}})$ at the reference $\bm{x}_{\mathrm{ref}}$ and the error $(L_f\!+\!L_{\tilde{\mu}})\|\bm{e}\|$ of linearization around the reference, the uncertainty in the dynamics can be bounded. Thereby, the region with decreasing Lyapunov function can be efficiently determined.}
		\label{fig:trackbound}
		
	\end{figure}

	In more detail, we consider feedback linearizable systems
	\begin{equation} \label{equ:FELIsystem}
	\dot{x}_{1}=x_{2}, \quad \dot{x}_{2}=x_{3}, \quad \ldots \quad \dot{x}_{d_x}=f(\boldsymbol{x})+g(\bm{x})u,
	\end{equation}
	where a scalar control input $u$ as well as scalar functions $f:\mathbb{X}\rightarrow\mathbb{R}$ and $g:\mathbb{X}\rightarrow \mathbb{R}$ are assumed only for simplicity of exposition, while all derived results straightforwardly extend to multi-input systems in the canonical form. 
	Similar to previous work \cite{0-6_umlauft2019feedback}, we assume that $f(\cdot)$ is an unknown function, while~$g(\cdot)$ is known. The knowledge of $g(\cdot)$ is merely used to streamline the presentation, but all results can be extended to unknown functions $g(\cdot)$ following the approach in \cite{ Lederer2020a}. In order to ensure global controllability of \eqref{equ:system}, the following assumption is needed.\looseness=-1
	\begin{assumption}\label{ass:g}
		The function $g(\cdot)$ is positive, i.e., $g(\bm{x})\!>\! 0 $.
	\end{assumption}
	This assumption is a standard condition when designing control laws for systems in the canonical form \cite[Definition 13.1]{0-7_khalil2002nonlinear}, and ensures the non-singularity of $g(\cdot)$. It is naturally satisfied by many systems such as Euler-Lagrange systems, where $g(\cdot)$ corresponds to the positive inertia. Therefore, this assumption is not restrictive in practice.
	
	Additionally, we assume that the unknown function is well-behaved, which is formalized in the following.
	\begin{assumption}\label{ass:f}
		The function $f(\cdot)$ is $L_f$-Lipschitz.
	\end{assumption}
	This assumption globally ensures a unique solution for the system \eqref{equ:system} \cite{0-7_khalil2002nonlinear}, such that it can be commonly found in control. Since it is satisfied by many systems such as Euler-Lagrange dynamics in practice, it is not restrictive. 
	
	In order to allow the accurate tracking of the reference trajectory with system \eqref{equ:FELIsystem}, we consider reference trajectories\looseness=-1
	\begin{align}
	\boldsymbol{x}_{\mathrm{ref}}\left(t\right)=\begin{bmatrix}x_{\mathrm{ref}}(t) & \dot{x}_{\mathrm{ref}}(t) &\cdots& \frac{\mathrm{d}^{d_x-1}}{\mathrm{d} t^{d_x-1}} x_{\mathrm{ref}}(t)\end{bmatrix}^{T},
	\end{align}
	where $x_{\mathrm{ref}}:\mathbb{R}\rightarrow\mathbb{R}$ is $d_x$ times continuously differentiable. For tracking this trajectory, we employ the feedback linearizing control law 
	\begin{equation} \label{equ:FL}
	u=\pi_{\mathrm{FL}}(\boldsymbol{x})=\frac{1}{g(\bm{x})}\left(-\tilde{\mu}(\boldsymbol{x})+\nu+\frac{\mathrm{d}^{d_x}}{\mathrm{d} t^{d_x}} x_{\mathrm{ref}}(t)\right),
	\end{equation} 
	where the mean $\tilde{\mu}(\bm{x})$ defined in \eqref{eq:log-gp mean} is used as model. The input $\nu$ to the approximately linearized system is given by the linear feedback law $\nu=-k_c\begin{bmatrix}\lambda_1&\cdots&\lambda_{d_x-1}&1 \end{bmatrix}\bm{e}$, 
	where $k_c\in\mathbb{R}_+$ is the control gain and $\lambda_1,\ldots,\lambda_{d_x-1}\in\mathbb{R}$ are coefficients such that for $s\in\mathbb{C}$, the polynomial $s^{d_x-1}+\lambda_{d_x-1}s^{d_x-2}+\ldots+\lambda_1$ is Hurwitz. Due to these choices, the error dynamics can be compactly expressed by\looseness=-1
	\begin{equation} \label{eq:edyn}
	\dot{\boldsymbol{e}}\!=\!\!\underbrace{\begin{bmatrix}
		\bm{0} & \bm{I}\\
		-\lambda_1 k_c& -\!\begin{bmatrix}
		\lambda_2 k_c&\cdots&k_c
		\end{bmatrix}
		\end{bmatrix}}_{\boldsymbol{A}}\!\! \boldsymbol{e}+(f(\boldsymbol{x})\!-\!\tilde{\mu}(\boldsymbol{x}))\!\!\begin{bmatrix}0\vspace{-0.15cm}\\\vdots\\1 \end{bmatrix}\!\!.\!\!
	\end{equation} 
%
	The matrix $\bm{A}$ defines a stable dynamical system because of the Hurwitz coefficients $\lambda_i$, which is independent of the online learning. Therefore, the second summand in \eqref{eq:edyn} can be considered a disturbance depending on the online learning, which allows us to straightforwardly analyze the ultimate boundedness of this system using Lyapunov theory.
	
	\begin{theorem} \label{the:main theorem}
		Consider a dynamical system~\eqref{equ:system}, where $f(\cdot)$ is a sample from a Gaussian process with $L_k$-Lipschitz kernel $k(\cdot,\cdot)$. Moreover, assume Assumptions~\ref{ass:data},~\ref{ass:g} and~\ref{ass:f} hold, and choose a positive definite, symmetric matrix $\bm{Q}\in\mathbb{R}^{d\times d}$ and a control gain $k_c$ such that 
		\begin{align}\label{eq:gain cond}
		\|\bm{p}_d(k_c)\|< \frac{\lambda_{\min}(\bm{Q})}{2(L_f+L_{\tilde{\mu}})},
		\end{align}
		where $\bm{P}\!=\![\bm{p}_1(k_c)\ \cdots\ \bm{p}_d(k_c)]$ is the solution to the Lyapunov equation $\bm{A}^T\bm{P}\!+\!\bm{P}\bm{A}\!=\!-\bm{Q}$ and $L_{\tilde{\mu}}$ denotes the Lipschitz constant of the LoG-GP mean $\tilde{\mu}(\cdot)$. Then, the online learning feedback linearizing control law \eqref{equ:FL} guarantees an ultimately bounded tracking error\looseness=-1
		\begin{align}\label{eq:ultbound}
		\!\vartheta(\kappa_t)\!=\!\frac{2\left\|\boldsymbol{p}_{d}(k_c)\right\|    \sqrt{\!\lambda_{\max}(\bm{P})}\max_{t'\!\in [0,t]}\eta\left(\boldsymbol{x}_{\mathrm{ref}}(t') \right)}{\left(\lambda_{\min}(\bm{Q})\!-\!2\left\|\boldsymbol{p}_{d}(k_c)\right\|\left(L_{f}\!+\!L_{\tilde{\mu}}\right) \right)\! \sqrt{\!\lambda_{\min}(\bm{P})}}.\!
		\end{align}
	\end{theorem}
	\vspace{0.3cm}
	\begin{proof}
		See \cref{app:D}.
	\end{proof}
	This theorem has the advantage over previously derived results on similar settings \cite{0-5_lederer2019uniform, 0-6_umlauft2019feedback} that the ultimate bound $\vartheta$ in \cref{the:main theorem} depends only on the standard deviation along the reference $\bm{x}_{\mathrm{ref}}$. Thereby, the ultimate bound \eqref{eq:ultbound} can be efficiently computed, whereas the results in previous works provide only an implicit representation of the ultimate bound due to the dependency of the GP error bound on the state $\bm{x}$. This advantage resulting from the linearization around the reference $\bm{x}_{\mathrm{ref}}$ in \eqref{eq:lin_xref} comes at the cost of the additional constraint \eqref{eq:gain cond} for the control gain~$k_c$ compared to existing approaches~\cite{0-5_lederer2019uniform, 0-6_umlauft2019feedback}. Even though this constraint makes the application of \cref{the:main theorem} more restrictive, this weakness is strongly outweighed by the benefit of the explicit tracking error bound for determining the active models online.
	
	\begin{remark}
		Due to the definition of $\bm{A}$ in \eqref{eq:edyn}, it can be straightforwardly checked that it is always possible to ensure the satisfaction of \eqref{eq:gain cond} for a fixed matrix $\bm{Q}$ by choosing a sufficiently large gain $k_c$. Therefore, condition~\eqref{eq:gain cond} effectively imposes a lower bound for the control gains~$k_c$ admitting ultimate tracking error bounds \eqref{eq:ultbound}.
	\end{remark}

\section{EVALUATION IN EXOSKELETON CONTROL}\label{sec:eval}
\setlength{\textfloatsep}{10pt}
\setlength{\floatsep}{13pt}

In order to evaluate the applicability of the proposed networked on-line learning approach for resource constrained systems\footnote{Open-source code conceptually demonstrating the proposed method is available at \url{https://gitlab.lrz.de/online-GPs/cloud-GPs}.}, we employ it for the control of an upper-limb human-exoskeleton assisting a user in tracking a reference trajectory, which is simulated in Julia \cite{Bezanson2017}, a modern programming language for accelerating physics simulations. Since the exoskeleton is intended to be used in a portable manner, this scenario resembles an example for a wearable robotic system with memory and computational constraints. These constraints are particularly challenging for the control of the exoskeleton as human user data is required in practice to infer models allowing for personalized assistance.

\begin{figure}
\vspace{0.15cm}
	\centering
	\includegraphics[width=0.25\textwidth]{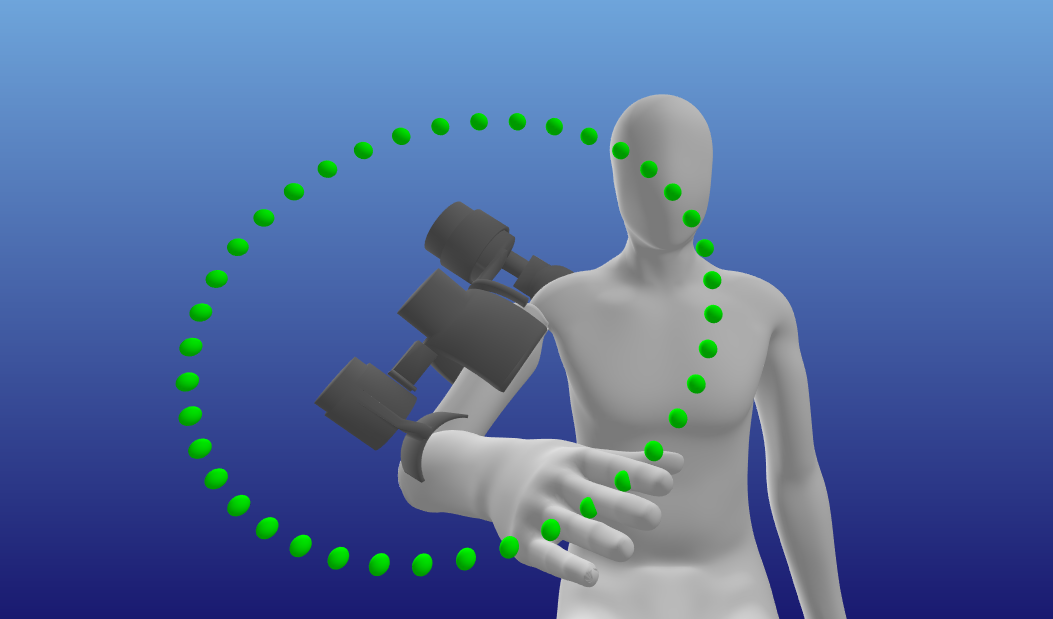}
	\vspace{-0.3cm}
	\caption{Visualization of the upper-limb human-exoskeleton simulation and trajectory tracking task. The green circles depict discrete points along the elliptic reference trajectory, which must be followed with the hand. }
	\label{fig:sim_vis}
\end{figure}

For the simulation, we assume a rigid kinematic coupling between the human and exoskeleton arm, which allows the modelling of both as one kinematic chain consisting of four DoFs. The exoskeleton model is based on the design described in \cite{Trigili2020}, whilst the model parameter for the human are chosen according to anthropometric tables \cite{drillis_body_1964}. Here, the reference is set to $70\,\si{\kilogram}$ and $1.75\,\si{\meter}$. As illustrated in \cref{fig:sim_vis}, the goal is to track an elliptic trajectory with the hand of the human by employing the learning-based feedback linearizing control law \eqref{equ:FL}. Each period of the ellipse takes $T_p=6\si{\second}$, the simulation runs at $1\si{\kilo\Hz}$, and we consider a memory constraint of $\bar{M}=4000$ data pairs for the local memory. Streaming data for online learning is generated with noise standard deviation $\sigma_\text{on}=0.05$ at a sampling rate of $100\si{\hertz}$, i.e., $\tau=10 \si{ms}$. Each local GP model can contain a maximum of $\bar{N}=100$ training points and the hyperparameters are set to $\sigma_f = 1$, $l_i=1$/$l_i=3$ for inputs corresponding to joint angles/angular velocities. Algorithm~\ref{alg:decrease} is run with temporal discretization $\Delta t = 10 \si{ms}$ and $N_s=1000$ random samples. Finally, the control gains are set to $k_c=400$ and $\lambda=1$.

\begin{table}[!t]
	\vspace{0.05cm}
	\caption{Network bandwidth $B$, time delay $T_d$ and resulting time when learning is stopped $T_s$}\label{tab:med_size}
	\centering
	\begin{scriptsize}
	\vspace{-0.35cm}
			\begin{tabular}{l c c c c}
				\toprule
				 & low & medium & high & large\\
				  & bandwidth & bandwidth & bandwidth & delay\\
				\midrule
				$B$ [$\si{samples/s}$] & $1500$ & $3000$ & $10000$ & $10000$\\
				$T_d$ [$\si{s}$] & $0.1$ & $0.1$ & $0.1$ & $1.0$\\
				$T_s$ [$\si{s}$] & $44.67$ & $60.04$ & --- & $30.81$\\
				\bottomrule
			\end{tabular}
		
	\end{scriptsize}
	
\end{table}

In order to investigate the dependency of the tracking accuracy and memory occupation on the network bandwidth $B$ and time delay $T_d$, we compare networked LoG-GP controllers with access to different network connections as outlined in \cref{tab:med_size}. 
Additionally, we employ 
a LoG-GP without memory constraints, i.e., $\bar{M}=\infty$, as baseline to demonstrate the absence of a performance loss of the networked LoG-GP when a sufficiently high bandwidth is available. 
The average update time for the LoG-GP is $0.3\si{ms}\!<\!\tau$ in all simulations, and the resulting curves for the evolution of the local memory occupation are depicted in \cref{fig:numData}. Since the LoG-GP has low accuracy during the first period, the tracking error bound $\vartheta(\kappa)$ is large during the first $6\si{s}$, such that all data is required on the local system. After this period, the different curves exhibit the behavior discussed in \mbox{\cref{subsec:4.3}}: The lower the bandwidth $B$, the faster the local memory consumption grows. Moreover, an increase in time delay $T_d$ causes a significantly faster growing memory occupation. Due to the limited local memory, this leads to an early stop in learning at the times depicted in \cref{tab:med_size},
after which the memory occupation stagnates.\looseness=-1

\begin{figure}
\vspace{0.15cm}
	\centering
	\includegraphics[width=0.47\textwidth]{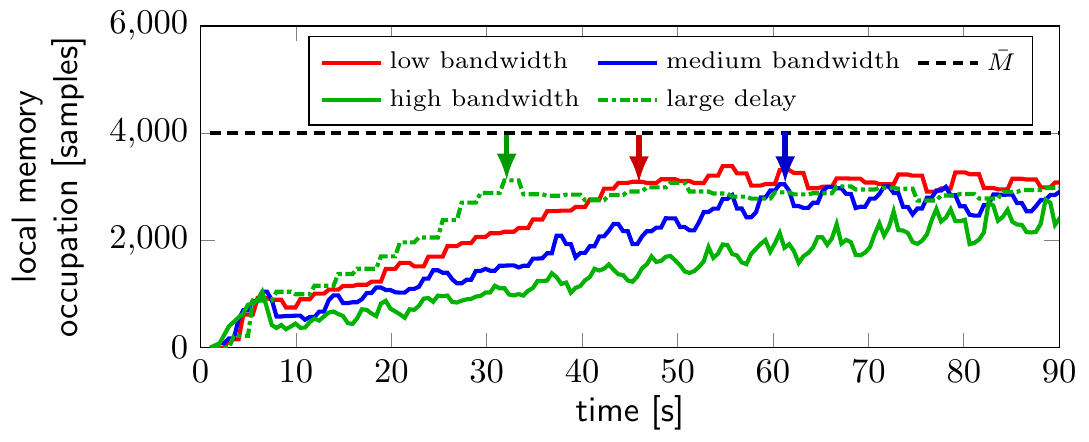}\vspace{-0.4cm}
	\caption{The higher the bandwidth $B$, the longer the LoG-GP can learn before the number of training pairs in the local memory reaches the limitations. Large time delay $T_d$ causes a significantly earlier stopping of learning, as indicated by the arrows. \looseness=-1}
	\label{fig:numData}
\end{figure}

\begin{figure}
	\centering
	\includegraphics[width=0.47\textwidth]{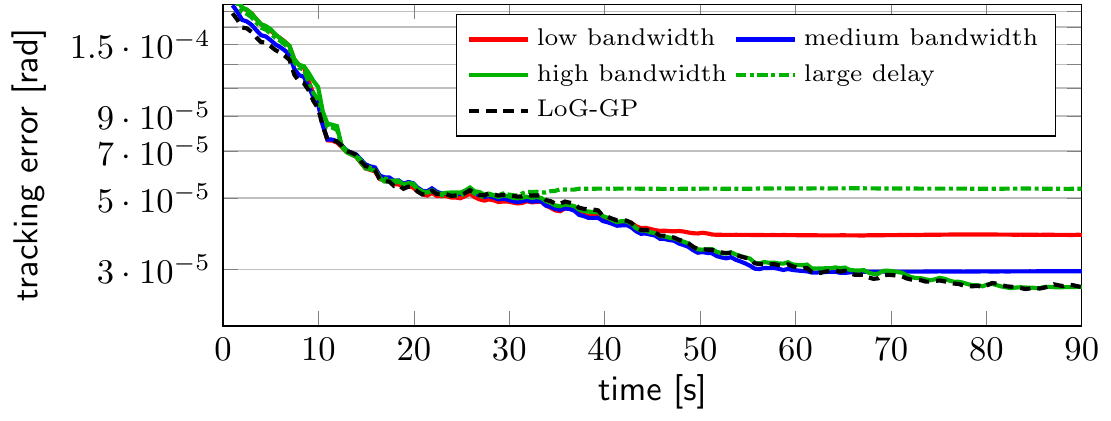}\vspace{-0.4cm}
	\caption{When the memory limitation is reached and the learning process stops, the tracking error stagnates. Since higher bandwidths $B$ allow learning for a longer time, larger values of $B$ yield lower tracking errors eventually. While learning is not stopped, networked LoG-GPs ensure the same tracking accuracy as LoG-GPs without any constraints. Overall, online learning significantly improves the tracking accuracy over the baseline case without model learning, which is not depicted since it permanently exceeds $2\cdot 10^{-2}.$\looseness=-1}
	\label{fig:tracking_error}
\end{figure}

This stagnation has an immediate effect on the evolution of the tracking error, as illustrated in \cref{fig:tracking_error}. While the feedback linearizing control law \eqref{equ:FL} with networked LoG-GP model achieves the same improvement in tracking accuracy as with the unconstrained LoG-GP when model updates are performed, the tracking performance ceases to improve and effectively remains constant after the learning has stopped. Due to the continual learning of the networked LoG-GP with a high bandwidth connection, the corresponding evolution of the tracking error is visually identical to the curve resulting from usage of the LoG-GP without memory constraint. This clearly demonstrates that the proposed approach allows a transfer of data to the cloud without any loss in performance when sufficient transmission bandwidth is available. Moreover, even when online learning has to be stopped early, it still yields a significant improvement in tracking accuracy compared to the baseline case without model learning, where a stationary error of $\approx 2\cdot 10^{-2}$ rad has been observed. This strongly underlines the advantages of online model inference for model-based control despite resource constraints.

\section{CONCLUSION}\label{sec:conc}
This paper presents a novel networked online learning approach for control of safety-critical systems with local resource constraints based on Gaussian process regression. By employing a tree-structured local GP approximation, relevant local models for control can be efficiently determined in a sampling-based fashion. This is exploited in the design of an effective data transmission scheme, which ensures the timely availability of data in the local computing unit. The effectiveness of the proposed networked online learning approach is demonstrated in a simulation of a robotic exoskeleton.

\appendix

\subsection{Proof of \cref{th:sample_active}}\label{app:A}
\begin{proof}
			Due to \cref{ass:errbound}, at each time $t$, the tracking error $\bm{e}$ is bounded by $\vartheta(\kappa_t)$. It is straightforward to see that $\bm{x}_{\mathrm{ref}}(\cdot)$ is Lipschitz continuous, such that \looseness=-1
			\begin{align}
			\!\|\bm{e}(t)\|\leq \xi\!-\!2\zeta ~~ \forall t\!\in \!\Big[t_1\!+\!\frac{2j\!-\!1}{2}\Delta t,t_1\!+\!\frac{2j\!+\!1}{2}\Delta t\Big]\!
			\end{align}
			and consequently $\mathbb{T}_{t_1}^{t_2}\subset \bigcup_{j=1}^{\lceil\frac{t_2-t_1}{\Delta t} \rceil}\mathbb{B}_{\xi-2\zeta}$. 
			Therefore, it remains to show that the set of active models for time $t_1+j\Delta t$ defined as $\mathbb{A}_{t_1+j\Delta t}=\bigcup_{\bm{x}\in \mathcal{B}_{\xi-2\zeta}}\mathbb{A}_{\bm{x}}$ 
			is overapproximated by Algorithm~\ref{alg:decrease}. For this purpose, choose any model $l\in\mathbb{A}_{t+j\Delta t}$. Then, the intersection between the active region $\mathbb{X}_l$ of this model and the ball $\mathcal{B}_{\xi}$ has a volume of at least $ \pi^{d_x/2}(\min\{r_{\min},\zeta\})^{d_x}/\Gamma(\frac{d_x}{2}+1)$, where $\Gamma:\mathbb{R}_+\rightarrow\mathbb{R}_+$ denotes Euler's gamma function
			. Therefore, the probability of a sample $\bm{x}^{(i)}\sim\mathcal{U}(\mathcal{B}_{\xi})$ being in the active region of model $l$ can be bounded by
			\begin{align}
			P(\omega_l(\bm{x}^{(i)})>0|l\in\mathbb{A}_{t_1+j\Delta t})\geq \frac{\min\{r_{\min}^{d_x},\zeta^{d_x}\}}{\xi^{d_x}}.
			\end{align}
			The probability of none of the $N_s$ samples falling into the active region $\mathbb{X}_l$ is consequently upper bounded by $(1\!-\!P(\omega_l(\bm{x}^{(i)})\!>\!0|l\!\in\!\mathbb{A}_{t_1+j\Delta t}))^{N_s}$,
			such that \eqref{eq:activeModelProb} follows from the union bound over all time steps and all models $l\in\mathbb{L}$.
\end{proof}

\subsection{Proof of \cref{lem:netcon}}\label{app:B}
\begin{proof}
		Satisfaction of the memory constraint~\eqref{eq:memcon} implies that the transmission of $\mathbb{D}_{j}$, $j\!\in\!\mathbb{N}$, can be achieved with time~$T_{\mathrm{trans}}\!\leq\! \frac{\bar{M}}{2B}\!+\!T_d$. 
		Hence, we have~$T_{\mathrm{access}}\!=\!\frac{\Delta T}{2}\!\geq\! T_{\mathrm{trans}}$, guaranteeing satisfaction of the network constraint \eqref{eq:netcon}. \looseness=-1
\end{proof}

\subsection{Proof of \cref{th:transscheme}}\label{app:C}
\begin{proof}
		Since the cardinality of $\mathbb{D}_{j+q}$ can be bounded by $|\mathbb{D}_{j+q}|\leq |\mathbb{D}_{j}|+\bar{m}$, memory constraints are satisfied as long as $|\mathbb{D}_{j}|\leq \bar{M}/2-\bar{m}$. Therefore, $\iota$ as defined in \eqref{eq:iota} ensures that the memory constraint \eqref{eq:memcon} is satisfied, which implies the satisfaction of the network constraint \eqref{eq:netcon} due to \cref{lem:netcon}.\looseness=-1
\end{proof}

\subsection{Proof of \cref{the:main theorem}}\label{app:D}
\begin{proof}
		In order to prove the ultimate bound, we employ the Lyapunov function $V(\boldsymbol{e})\!=\!\boldsymbol{e}^{T} \boldsymbol{P} \boldsymbol{e}$, where $\bm{P}$ is a positive definite matrix as $\bm{\lambda}$ is a Hurwitz vector. 
		The derivative of the Lyapunov function is guaranteed to satisfy\looseness=-1
		\begin{align}
		\dot{V}(\boldsymbol{e})
		&=-\boldsymbol{e}^{T}\bm{Q}\boldsymbol{e}+2 \boldsymbol{e}^{T} \boldsymbol{p}_{d}(k_c)\left(f(\boldsymbol{x})-\tilde{\mu}(\boldsymbol{x})\right).
		\end{align}
        Due to Lipschitz continuity, we obtain
		\begin{align}
		\dot{V}(\boldsymbol{e}) 
		& \!\leq\!-\lambda_{\min}\!(\bm{Q})\|\boldsymbol{e}\|^{2}\!+\!2\|\boldsymbol{e}\|\!\left\|\boldsymbol{p}_{d}(k_c)\right\|\!\left(L_{f}\!+\!L_{\tilde{\mu}}\right)\!\left\|\boldsymbol{x}\!-\!\boldsymbol{x}_{\mathrm{ref}}\right\|\notag \\
		&  \quad\! +\!2\|\boldsymbol{e}\|\left\|\boldsymbol{p}_{d}(k_c)\right\|\left|\left(f\!\left(\boldsymbol{x}_{\mathrm{ref}}\right)\!-\!\tilde{\mu}\left(\boldsymbol{x}_{\mathrm{ref}}\right)\right)\right|,
		\label{eq:lin_xref}
		\end{align}
		where the Lipschitz constant $L_{\tilde{\mu}}$ in \eqref{eq:gain cond} follows directly from Lipschitz continuity of the individual mean functions resulting from the $L_k$-Lipschitz kernel $k(\cdot,\cdot)$~\cite{0-5_lederer2019uniform}. Due to \cref{the:Regression Error Bound}, the error between the unknown function $f(\cdot)$ and the LoG-GP mean $\tilde{\mu}(\cdot)$ can be bounded, such that we obtain
		\begin{align} 
		\dot{V}(\boldsymbol{e}) 
		& \leq -\left(\lambda_{\min}(\bm{Q})-2\left\|\boldsymbol{p}_{d}(k_c)\right\|\left(L_{f}+L_{\tilde{\mu}}\right)\right)\|\boldsymbol{e}\|^{2} \notag \\
		&  \quad\, +2\|\boldsymbol{e}\|\left\|\boldsymbol{p}_{d}(k_c)\right\| \eta\left(\boldsymbol{x}_{\mathrm{ref}}\right).
		\end{align}
		Due to \eqref{eq:gain cond}, the Lyapunov derivative is negative for 
		\begin{align}\label{eq:Lyapdecr}
		\|\boldsymbol{e}\| > \frac{2\left\|\boldsymbol{p}_{d}(k_c)\right\|   \eta\left(\boldsymbol{x}_{\mathrm{ref}}\right)}{\lambda_{\min}(\bm{Q})-2\left\|\boldsymbol{p}_{d}(k_c)\right\|\left(L_{f}+L_{\tilde{\mu}}\right)}.
		\end{align}
		Since the ultimately bounded set is given by the smallest sub-level set of $V(\cdot)$ which contains the ball defined through \eqref{eq:Lyapdecr}, we can directly determine it as $\{\bm{e}:V(\bm{e})\leq \vartheta(\kappa_t)^2\lambda_{\min}(\bm{P})\}$ due to the quadratic structure of $V(\cdot)$. Over-approximating this set by a ball concludes the proof.
\end{proof}
	

%


\bibliographystyle{IEEEtran}
\bibliography{Bib}

\end{document}